\pdfoutput=1
\pdfpagewidth=\paperwidth
\pdfpageheight=\paperheight
\documentclass[12pt, letterpaper]{article}

\DeclareFixedFont{\MyTitleFont}{OT1}{ptm}{m}{n}{20pt}
\DeclareFixedFont{\MyAuthorFont}{OT1}{ptm}{m}{n}{13pt}
\DeclareFixedFont{\MyAbstractTitleFont}{OT1}{ptm}{m}{it}{12pt}
\DeclareFixedFont{\MyAbstractFont}{OT1}{ptm}{m}{it}{11pt}
\DeclareFixedFont{\MySubtitleFont}{OT1}{ptm}{m}{n}{15pt}
\DeclareFixedFont{\MySubSubtitleFont}{OT1}{ptm}{m}{n}{13pt}
\DeclareFixedFont{\MySubSubSubtitleFont}{OT1}{ptm}{m}{n}{11pt}
\DeclareFixedFont{\MyTextFont}{OT1}{ptm}{m}{n}{11pt}

\newcommand{\neweq}[1]{\stackrel{\smash{\footnotesize\mathrm{#1}}}{=}}

\newcommand{\newleq}[1]{\stackrel{\smash{\footnotesize\mathrm{#1}}}{\leq}}

\usepackage{geometry}
\usepackage{titling}

\title{\MyTitleFont Sensitivity Analysis of Continuous-Time Linear Control Systems subject to Control and Measurement Noise: An Information-Theoretic Approach \vspace{0em}}

\author{\MyAuthorFont Neng Wan$^1$, Dapeng Li$^2$, and Naira Hovakimyan$^1$}
\date{}

\usepackage{titlesec}

\usepackage[marginal]{footmisc}
\usepackage{abstract}

\usepackage[sort]{cite}

\usepackage{indentfirst}
\usepackage{amsthm}
\usepackage{amsfonts}
\usepackage{amsmath}
\usepackage{booktabs}
\usepackage{bm}
\usepackage{mathptmx}
\usepackage{subfigure}
\usepackage{enumerate}
\usepackage[mathcal]{euscript}
\usepackage{float}
\usepackage{graphics} 
\usepackage{epsfig} 

\usepackage[pdftex, bookmarks]{hyperref}
\hypersetup{colorlinks = true, citecolor = red, linkcolor = blue, urlcolor = black}

\newtheorem{theorem}{Theorem}
\newtheorem{remark}{Remark}

\newtheorem{lemma}{Lemma}

\newtheorem{property}{Property}
\newtheorem{definition}{Definition}
\newtheorem{corollary}[theorem]{Corollary}

\newtheorem{manualtheoreminner}{Theorem}
\newenvironment{manualtheorem}[1]{%
	\renewcommand\themanualtheoreminner{#1}%
	\manualtheoreminner
}{\endmanualtheoreminner}

\allowdisplaybreaks

\begin{document}

\maketitle


\footnotetext[0]{$^{1}$Neng Wan and Naira Hovakimyan are with the Department of Mechanical Science and Engineering, University of Illinois at Urbana-Champaign, Urbana, IL 61801, USA. {\tt\small \{nengwan2, nhovakim\}@illinois.edu}.}
\footnotetext[0]{$^{2}$Dapeng Li is a Principal Scientist with the JD.com Silicon Valley Research Center, Santa Clara, CA 95054, USA. {\tt\small dapeng.li@jd.com}.}
\vspace{-4em}

\begin{abstract}
{\vspace{0.5em}
	Sensitivity of linear continuous-time control systems, subject to control and measurement noise, is analyzed by deriving the lower bounds of Bode-like integrals via an information-theoretic approach. Bode integrals of four different sensitivity-like functions are employed to gauge the control trade-offs. When the signals of the control system are stationary Gaussian, these four different Bode-like integrals can be represented as differences between mutual information rates. These mutual information rates and hence the corresponding Bode-like integrals are proven to be bounded below by the unstable poles and zeros of the plant model, if the signals of the control system are wide-sense stationary.
}
\end{abstract}
\vspace{-0.5em}

\titleformat*{\section}{\centering\MySubtitleFont}
\titlespacing*{\section}{0em}{1.25em}{1.25em}[0em]

\section{Introduction}\label{sec1}

Stabilization of systems subject to external disturbances and achieving desired level of performance have been the objective of  feedback synthesis since its inception \cite{Bode_1945, Francis_TAC_1984, Zhou_1998, Hovakimyan_2012, Astrom_2013, Shtessel_2014}. With the visible progress of information technologies and their applications to feedback control systems over the last two decades, a great deal of attention has been given to understanding the fundamental limitations of closed-loop systems in the presence of communication channels, \cite{Zang_SCL_2003, Martin_TAC_2008, Ishii_SCL_2011, Li_TAC_2013, Fang_2017}. The main contribution of these papers was to explore the performance limitations of stochastic systems in the presence of limited information. While~\cite{Iglesias_LAA_2002, Zang_SCL_2003, Martin_TAC_2008, Ishii_SCL_2011} investigated the Bode-like integrals for discrete-time systems by using  Kolmogorov's entropy-rate equality~\cite{Cover_2012}, the results in~\cite{Li_TAC_2013} put forward an approach to explore the continuous-time systems by resorting to mutual information rates. In the aforementioned papers, sensitivity-like functions were introduced to define the Bode-like integrals, which can be regarded as a generalization of the classical Bode integrals from the deterministic LTI systems to stochastic nonlinear systems. For deterministic LTI systems, previous results based on complex analysis have shown that the lower bounds of the Bode integrals are determined by the unstable zeros and poles of the plant model~\cite{Bode_1945, Freudenberg_1985, Freudenberg_1987, Sung_IJC_1988, Middleton_1991}. Seminal results on this topic were also reported in~\cite{Middleton_2010, Seron_2012, Li_Auto_2013, Zhao_Auto_2014, Yu_IJRNC_2015}.

Performance limitations of stochastic systems in the presence of limited information were analyzed through the sensitivity-like functions~\cite{Martin_TAC_2007, Martin_TAC_2008, Fang_TAC_2017, Li_TAC_2013, Okano_Auto_2009, Ishii_SCL_2011}, which are defined by the power spectral densities (PSDs) of signals. Taking an information-theoretic approach was the key to get Bode integrals extended to stochastic nonlinear systems. Unlike the frequency-domain approach, which explicitly depends on the input-output relationship of the feedback system (transfer function), the focus of the information-theoretic approach is on the signals. The lower bound for sensitivity Bode-like integral in stochastic continuous-time systems was first reported in~\cite{Li_TAC_2013}: $\frac{1}{2\pi} \int_{-\infty}^{\infty} \log |T_{uw}(\omega)| d\omega \geq \sum_{\lambda \in \mathcal{UP}} p_i$. This result can be applied to  systems with nonlinear controllers, which is an improvement upon the prior results based on the frequency-domain approaches \cite{Bode_1945, Freudenberg_1985, Freudenberg_1987, Middleton_1991, Sung_IJC_1989, Seron_2012, Zhou_1998, Sung_IJC_1988}. Nevertheless, to the best of the authors' knowledge i) the lower bound of complementary sensitivity Bode-like integral in stochastic systems and ii) the lower bounds of load disturbance sensitivity and noise sensitivity Bode-like integrals in both deterministic and stochastic systems have rarely been investigated. Unboundedness of these sensitivity-like functions in high frequencies, as well as the challenges in information-theoretic representations of weighted Bode-like integrals have been the main obstacles on this path.

In this paper, a comprehensive sensitivity analysis of stochastic continuous-time systems, subject to control and measurement noise, is pursued via an information-theoretic approach. The continuous-time Bode-like integrals of sensitivity, complementary sensitivity, load disturbance sensitivity, and noise sensitivity are defined, where the Bode-like integrals for the latter two sensitivity functions were seldom studied before. Because Kolmogorov's entropy-rate equality, which was widely employed to derive the lower bounds of Bode-like integrals in discrete-time systems~\cite{Martin_TAC_2008, Ishii_SCL_2011, Fang_TAC_2017}, is not applicable to continuous-time systems, we resort to a seminal lemma on mutual information rates~\cite[p.~181]{Pinsker_1964} to seek  the lower bounds in continuous-time systems. With this lemma, we prove that the Bode-like integrals can be bounded below by the mutual information rates of signals, when the signals are stationary Gaussian. Furthermore, in the case when the signals of the system are wide-sense stationary, the lower bounds of these mutual information rates, and hence the lower bounds of Bode-like integrals, are determined by the unstable poles and zeros of the plant model. We also provide the relationship between Bode integrals and Bode-like integrals of continuous-time systems, which complements the previous investigations on discrete-time systems~\cite{Martin_TAC_2008, Ishii_SCL_2011}. Some early results in this direction were reported in~\cite{Wan_CDC_2018}.

The paper is organized as follows: \hyperref[sec2]{Section~2} introduces preliminary results and defines four different Bode-like integrals that gauge the performance limitations of continuous-time feedback systems; \hyperref[sec3]{Section~3} derives the lower bounds of Bode-like integrals in terms of mutual information rates; \hyperref[sec4]{Section~4} shows that these lower bounds are further bounded below by the unstable poles and zeros of the plant model; and \hyperref[sec5]{Section~5} concludes the paper.

\textit{Notations}. The notations used throughout this paper are defined as follows. In order to align with the notations used in control systems, $x(t)$ represents a continuous-time stochastic process with $x_{t_1}^{t_2}$ indicating a sample path on an interval $[t_1, t_2] \subset \mathbb{R}^{+}$ and $x^t := x_0^t$; $x(k)$ denotes a discrete-time stochastic process with $x_m^n$ indicating the segment $\{ x(k) \}_{k=m}^n$, $m<n \in \mathbb{N}$ and $x_0^n := x^n$; and $x^{(\delta)}$ denotes the discrete-time process obtained from sampling of $x(t)$ with an interval $\delta > 0$ with $x_i^{(\delta)} = x^{(\delta)}(i) := x(t_0 + i\delta), i \in \mathbb{N}$. We also use the information-theoretic notations: $\mathbb{E}[\cdot]$ for expectation, $h({x})$ for Shannon differential entropy, $I(\cdot; \cdot)$ for mutual information, and $I_\infty(\cdot; \cdot)$ for mutual information rate. The logarithmic function $\log(\cdot)$ in this paper assume the basis $\mathrm{e}$ by default. For a matrix $M$, $|M|$ denotes its determinant.

\section{Preliminaries and Problem Formulation}\label{sec2}

Consider the general continuous-time feedback system shown in \hyperref[fig1]{Figure~1},

\begin{figure}[H]\label{fig1}
	\centering \vspace{-2em}
	\includegraphics[width=0.45\textwidth]{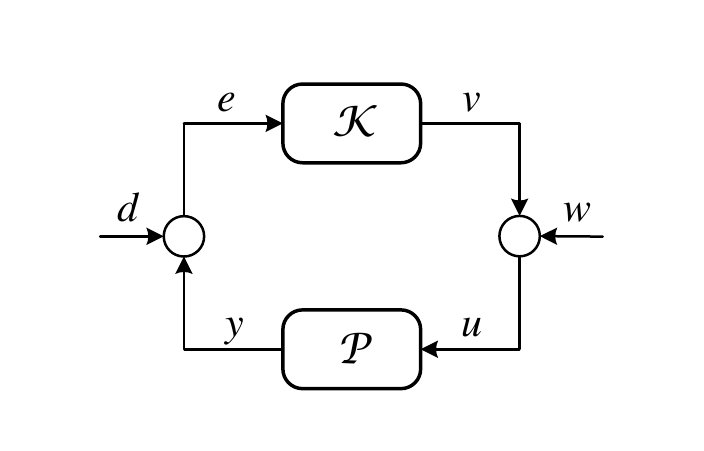} \vspace{-2em}
	\caption{\small General feedback system.}
\end{figure}

\noindent where $\mathcal{P}$ is the plant model, $\mathcal{K}$ denotes the causal feedback control mapping, and $d$ and $w$ respectively represent the noise over measurement and control channels. In classical control theory, with zero initial condition, the linear plant $\mathcal{P}$ can be described by the following transfer function
\begin{equation}\label{original_tf}
G(s) = c \cdot \dfrac{\prod_{i=1}^{m}(s-z_i)}{\prod_{i=1}^n(s-p_i)},
\end{equation}
where $c\in\mathbb{R}$, $m \leq n$ so that the system is causal, and $z_i$ and $p_i$ denote the zeros and poles of plant $\mathcal{P}$, respectively. When the control mapping $\mathcal{K}$ is linear, we use $C(s)$ to denote its transfer function, which has a similar form as~\eqref{original_tf}. In this scenario, the four important transfer functions for sensitivity analysis are given by
\begin{equation}\label{SF1}
	\begin{split}
		T_{uw}(s)  = \dfrac{1}{1+G(s)C(s)}, & \qquad T_{yw}(s)  = \dfrac{G(s)}{1+G(s)C(s)},\\
		T_{ud}(s)  = \dfrac{C(s)}{1+G(s)C(s)},& \qquad T_{yd}(s)  = \dfrac{G(s)C(s)}{1+G(s)C(s)},
	\end{split}
\end{equation}
where $T_{uw}(s), T_{yw}(s), T_{ud}(s)$ and $T_{yd}(s)$ respectively denote the sensitivity, load disturbance sensitivity, noise sensitivity, and complementary sensitivity function and are referred to as {\em Gang of Four} in~\cite{Astrom_2010}. For brevity, we also use the notation $L(s) = G(s) C(s)$ in the following context. The integral of sensitivity function $T_{uw}(s)$ over all frequencies, $(2\pi)^{-1} \cdot \int_{-\infty}^{\infty} \log \left|T_{uw}(s)\right| d\omega$, is known as the Bode's integral~\cite{Bode_1945, Seron_2012}, which is a critical index characterizing the performance limitations of feedback systems subject to noise. However, it seems that simply replacing the $T_{uw}(s)$ function in Bode integral with the other three sensitivity functions cannot give us a new trade-off as significant as the sensitivity Bode integral, due to the unboundedness of integrands $\log |T_{yw}(s)|$, $\log |T_{ud}(s)|$, and $\log |T_{yd}(s)|$ when $s \rightarrow \infty$ and the absence of a relationship between $T_{yw}(s)$ and $T_{ud}(s)$ that is similar to $T_{uw}(s) + T_{yd}(s) = 1$. A feasible solution dealing with the unboundedness of integrals is to multiply the logarithmic integrands by a weighting factor, such as $1/\omega^2$~\cite{Middleton_1991} or a Poisson-type kernel function~\cite{Seron_2012}.

In stochastic setting, the linear plant model $\mathcal{P}$ in \hyperref[fig1]{Figure~1} can be described by the following state-space model
\begin{equation}\label{original_sp}
	\begin{cases}
	\dot{x} = Ax + Bu,\\
	y  = Cx,
	\end{cases}
\end{equation}
where $x$ is the state vector, $u$ and $y$ are the input and output of $\mathcal{P}$. Throughout this paper, we assume that the initial values of the state vector are unknown, but have finite entropy. The measurement noise $d$ and control noise $w$  in~\hyperref[fig1]{Figure~1} are assumed to be mutually independent and zero-mean Gaussian. A comprehensive discussion on  different initial conditions for deterministic systems and stochastic systems is available in~\cite{Ishii_SCL_2011}. Derived from the sensitivity functions in~\eqref{SF1}, the following sensitivity-like functions in terms of PSDs were adopted by later researchers when analyzing the sensitivity properties of stochastic nonlinear systems via information-theoretic methods~\cite{Martin_TAC_2007, Ishii_SCL_2011, Li_TAC_2013}:
\begin{equation}\label{SF2}
	\begin{split}
	T_{uw}(\omega)  = \sqrt{\dfrac{\phi_u(\omega)}{\phi_w(\omega)}}, & \qquad T_{yw}(\omega)  = \sqrt{\dfrac{\phi_y(\omega)}{\phi_w(\omega)}},\\
	T_{ud}(\omega)  = \sqrt{\dfrac{\phi_u(\omega)}{\phi_d(\omega)}},& \qquad T_{yd}(\omega)  = \sqrt{\dfrac{\phi_y(\omega)}{\phi_d(\omega)}},
	\end{split}
\end{equation}
\noindent where each pair of signals is stationary and stationary correlated, $\phi_x(\omega)$ denotes the PSD of stationary signal $x(t)$ with
\begin{equation}\label{PSD}
	\phi_x(\omega) = \int_{-\infty}^{\infty} r_x(\tau) \cdot  {\rm e}^{-j\omega\tau} d\tau,
\end{equation}
and $r_x(\tau) = r_{xx}(t + \tau,t)$ is the auto-covariance of signal $x$ with $r_{xy}(v,t) = \textrm{Cov}[x(v), y(t)]$. Inspired by the weighting factor used for the complementary sensitivity Bode integral in~\cite{Middleton_1991}, in this paper we are interested in seeking the lower bounds for the Bode-like integrals defined as follows
\begin{equation}\label{Bode-Like}
	\begin{split}
	\dfrac{1}{2\pi}\int_{-\infty}^{\infty} \log T_{uw}(\omega) \ d\omega,  & \qquad \dfrac{1}{2\pi}\int_{-\infty}^{\infty} \log T_{yw}(\omega) \ {d\omega},\\
	\dfrac{1}{2\pi}\int_{-\infty}^{\infty} \log T_{ud}(\omega) \ \dfrac{d\omega}{\omega^2},  & \qquad \dfrac{1}{2\pi}\int_{-\infty}^{\infty} \log T_{yd}(\omega) \ \dfrac{d\omega}{\omega^2}.\\
	\end{split}
\end{equation}
\begin{remark}
	In this paper, the integrals on sensitivity functions defined in~\eqref{SF1} are referred to as Bode integrals, while the integrals on sensitivity-like functions defined in~\eqref{SF2} are named as Bode-like integrals. As~\hyperref[lem1]{Lemma~1} in this paper will reveal, the Bode integrals and the reciprocal Bode-like integrals are equal when the signals of system are stationary Gaussian. Hence, according to~\cite{Middleton_1991}, the weighting function $1/\omega^2$ should also guarantee the boundedness of the Bode-like integral of $T_{yd}$. Meanwhile, the forms of the Bode-like integrals on $T_{yw}$ and $T_{ud}$ are respectively inherited from the definitions of the sensitivity Bode integral on $T_{uw}$ and the complementary sensitivity integral on $T_{yd}$. This allows to preserve some good properties and also bring convenience when deriving the lower bounds. However, the absence of a similar identity as $T_{uw} + T_{yd} = 1$ between $T_{yw}$ and $T_{ud}$ and the involvement of weighting function $1 / \omega^2$ generate additional challenges and constraints when deriving the lower bounds via information-theoretic approaches.
\end{remark}

Before we proceed to discuss the lower bounds of Bode-like integrals, some basic definitions, properties from information theory, and a lemma on the relationship between Bode integrals and Bode-like integrals are given in the remaining part of this section. Motivated by the frequency inversion adopted in~\cite{Middleton_1991}, consider the following frequency transformation
\begin{equation}\label{freq_trans}
	\tilde{s} = j\tilde{\omega} = (j\omega)^{-1} = s^{-1},
\end{equation}
where the frequencies satisfy $\tilde{\omega} = -(\omega)^{-1}$. Applying the frequency inversion~\eqref{freq_trans} to the plant model $G(s)$ in~\eqref{original_tf}, we refer to the resulting system as the {\em auxiliary system} and write its transfer function as
\begin{equation}\label{tf_auxiliary}
G(s) = G(\tilde{s}^{\ -1}) = c \cdot \tilde{s}^{n-m} \cdot \dfrac{\prod_{i=1}^m (1 - \tilde{s} \cdot z_i)}{\prod_{i=1}^{n} (1 - \tilde{s} \cdot p_i)} = \tilde{G}(\tilde{s}).
\end{equation}
The Laplace transform of the signal $\tilde{x}$ in the auxiliary system satisfies $X(s) = X(\tilde{s}^{\ -1}) = \int_{-\infty}^{\infty} x(\tau) \cdot \mathrm{e}^{-\tau \cdot \tilde{s}^{ \ -1}} d\tau = \tilde{X}(\tilde{s})$, where $x$ can be replaced by any signal shown in~\hyperref[fig1]{Figure~1}.
When the control mapping $\mathcal{K}$ is linear, we use the notation $\tilde L( \tilde s) = \tilde G(\tilde s) \cdot \tilde C(\tilde s)$, with $\tilde C(\tilde s)$ being the frequency inversion of $C(s)$. No direct or intermediate result in this paper will be derived from the auxiliary system, and the transfer function~\eqref{tf_auxiliary} is not causal with respect to $\tilde{s}$ when  $n - m > 0$.

 By swapping the input and output of the auxiliary system, the resulted system is referred to as the inverse system $(\tilde{\mathcal{P}}^{-1}, \tilde{\mathcal{K}}^{-1})$, which has the plant with the following transfer function
\begin{equation}\label{tf_inv_auxiliary}
{\tilde{G}}^{-1}(\tilde{s}) = \dfrac{1}{c } \cdot \dfrac{ \prod_{i=1}^{n} (1 - \tilde{s} \cdot p_i)}{  \tilde{s}^{n-m} \cdot \prod_{i=1}^m (1 - \tilde{s} \cdot z_i)} = \dfrac{\prod_{i=1}^{n}(-p_i)}{c \cdot \prod_{i=1}^{m}(-z_i)} \cdot \dfrac{\prod_{i=1}^{n} (\tilde{s} - p_i^{-1})}{\tilde{s}^{n-m} \cdot \prod_{i=1}^{m}(\tilde{s} - z_i^{-1})}.
\end{equation}
When $p_i = 0$, the corresponding term $\tilde{s} - p_i^{-1}$ vanishes in the numerator. Hence, the plant model~\eqref{tf_inv_auxiliary} is proper with respect to $\tilde{s}$. A minimal realization of this transfer function is described by
\begin{equation}
\begin{cases}
\dot{\hat x} = \hat{A}\hat{x} + \hat B \tilde{v},\\
\tilde{e}  = \hat C \hat{x}.
\end{cases}
\end{equation}
The block diagram of the inverse system is shown in~\hyperref[fig2]{Figure~2}.
\begin{figure}[H]\label{fig2}
	\centering \vspace{-2em}
	\includegraphics[width=0.5\textwidth]{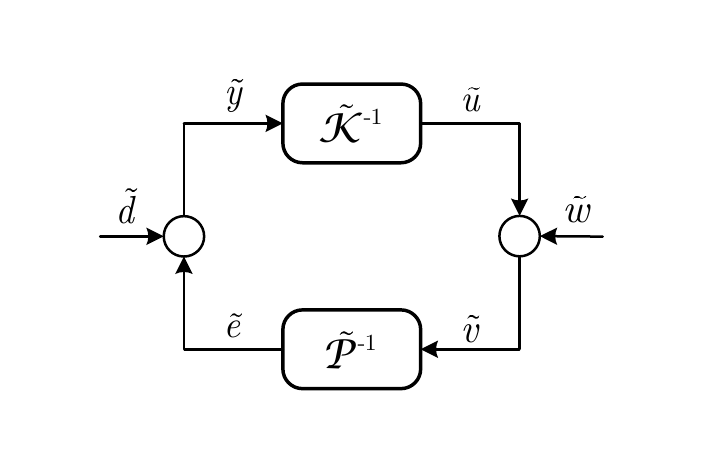} \vspace{-2em}
	\caption{\small Inverse system.}
\end{figure}

Some required definitions and  preliminary results are stated below, while more  information on these topics can be found in~\cite{Papoulis_2002, Gamal_2011, Cover_2012, Astrom_2012, Li_TAC_2013, Fang_2017}.
\begin{definition}\textnormal{({\textbf{Wide Sense Stationary}})}
	A second order random process $x(t)$ is called wide sense stationary, if $\mathbb{E}[x(t)] = \mathbb{E}[x(t+v)]$ and $\mathrm{Cov}[x(t), x(t + \tau)] = \mathrm{Cov}[x(v), x(v+\tau)]$.
\end{definition}

\begin{definition}\textnormal{({\textbf{Mean-Square Stability})}}
	A closed loop system is said to be mean-square stable, if the state $x(t)$ satisfies $\sup_{t \geq 0}\mathbb{E}\left[x^\mathrm{T}(t)x(t)\right] < \infty$.
\end{definition}

\begin{definition}\label{def3} \textnormal{({\textbf{Class $\mathbb{F}$ Function:  \textnormal{{See}\cite{Li_TAC_2013} {{or}} \cite[p.~182]{Pinsker_1964}}}})}  We define class $\mathbb{F}$  function in the following way: $\mathbb{F} = \{l:l(\omega) = p(\omega)(1-\varphi(\omega)), l(\omega) \in \mathbb{C}, \omega \in \mathbb{R} \}$, where $p(\cdot)$ is rational and $\varphi(\cdot)$ is a measurable function, such that $0 \leq \phi \leq 1$ for all $\omega \in \mathbb{R}$ and $\int_{\mathbb{R}}|\log(1-\varphi(\omega))| d\omega < \infty$.
\end{definition}

\begin{property}\label{ppt0}
	\textnormal{({See \cite{Li_TAC_2013}})} For two continuous-time stochastic processes $x$ and $y$, the mutual information between $x_{t_1}^{t_2}$ and $y_{t_1}^{t_2}$, $0 \leq t_1 < t_2 < \infty$, can be obtained as $I( x_{t_1}^{t_2}; y_{t_1}^{t_2}) = \lim_{n \rightarrow \infty}I(x_0^{(\delta(n))}, \cdots, x_n^{(\delta(n))}; \allowbreak y_0^{(\delta(n))}, \cdots, y_n^{(\delta(n))})$ for any fixed $t_1$ and $t_2$ with $x_i^{(\delta(n))} = x(t_1 + i \delta(n))$ and $\delta(n) = (t_2 - t_1) / (n+1)$.
\end{property}

\begin{property}\label{ppt1}
	\textnormal{({See \cite{Gamal_2011, Cover_2012, Fang_2017}})} For a pair of random variables ${x}$ and ${y}$, we have $h(x|y) = h(x + g(y)|y)$ and $I({x};{y}) = h({x}) - h({x}|{y}) = h({y}) - h({y}|{x}) = I(y; x)$, where $g(\cdot)$ is a measurable function.
\end{property}

\begin{property}\label{ppt2}
	\textnormal{(See \cite[p.~660]{Papoulis_2002})} For continuous random variables, $x_1, \cdots, x_n$, if the transformation $y_i = g_i(x_1, \cdots, x_n)$ has a unique inverse, then $h(y_1, \cdots, y_n) = h(x_1, \cdots, x_n) + \mathbb{E}\left[\log \left|J(x_1, \cdots, x_n)\right|\right]$, where $J(x_1, \cdots, x_n)$ is the Jacobian matrix of the above transformation.
\end{property}

\begin{property}\label{ppt3}
	\textnormal{(\textbf{Maximum Entropy}: See~\cite{Gamal_2011, Li_TAC_2013})} For a random vector $x \in \mathbb{R}^n$ with covariance matrix $\Sigma_x$, we have $h(x) \leq h(x_G) = 1 / 2 \cdot \log((2\pi e)^n \cdot |\Sigma_x|)$, where $x_G$ is Gaussian with the same covariance as $x$.
\end{property}

\begin{property}\label{ppt4}
	\textnormal{(See \cite{Li_TAC_2013} or \cite[p.~181]{Pinsker_1964})} Suppose that two one-dimensional continuous-time random processes $x(t)$ and $y(t)$ form a stationary Gaussian process $(x, y)$, and $\phi_x$ and $\phi_y$ belong to  class $\mathbb{F}$. Then, $I_\infty(x; y) = -(4\pi)^{-1}\int_{-\infty}^{\infty}\log\{1 - |\phi_{xy}(\omega)|^2/ [\phi_x(\omega)\phi_y(\omega)]\} d\omega$.
\end{property}


The following lemma establishes a relationship between the Bode integrals defined by the transfer functions and the Bode-like integrals defined by the PSDs of signals.
\begin{lemma}\label{lem1}
	\textit{When the plant model $\mathcal{P}$ and the controller $\mathcal{K}$ are linear, and the control noise $w(t)$ is wide sense stationary, the Bode integrals and the Bode-like integrals satisfy}
	\begin{subequations}
		\begin{equation}\label{lem1a}
		\dfrac{1}{2\pi}\int_{-\infty}^{\infty} \log {T}_{uw}(\omega) d\omega = \dfrac{1}{2\pi}\int_{-\infty}^{\infty} \log |T_{uw}(s)|  d\omega,
		\end{equation}
		\begin{equation}\label{lem1b}
		\frac{1}{2\pi} \int_{-\infty}^{\infty} \log{T}_{yw}(\omega)  {d\omega} = \frac{1}{2\pi} \int_{-\infty}^{\infty} \log |T_{yw}(s)|  {d\omega}.
		\end{equation}
		\textrm{When the measurement noise $d(t)$ is wide sense stationary, the Bode integrals and the Bode-like integrals satisfy}
		\begin{equation}\label{lem1c}
		\frac{1}{2\pi} \int_{-\infty}^{\infty} \log{T}_{ud}(\omega)  \frac{d\omega}{\omega^2} = \frac{1}{2\pi} \int_{-\infty}^{\infty} \log |T_{ud}(s)|  \dfrac{d\omega}{\omega^2},
		\end{equation}
		\begin{equation}\label{lem1d}
		\frac{1}{2\pi} \int_{-\infty}^{\infty} \log{T}_{yd}(\omega)  \frac{d\omega}{\omega^2} = \frac{1}{2\pi} \int_{-\infty}^{\infty} \log |T_{yd}(s)|  \dfrac{d\omega}{\omega^2}.
		\end{equation}
	\end{subequations}
\end{lemma}
\begin{proof}
	The proof is given in \hyperref[app]{Appendix}.
\end{proof}

\begin{remark}
	Since \hyperref[lem1]{Lemma~1} implies that the magnitudes of Bode integrals equal the magnitudes of the corresponding Bode-like integrals when disturbances are wide sense stationary, it is reasonable  to infer that the lower bounds of some Bode-like integrals defined in~\eqref{Bode-Like} should be identical to the lower bounds of Bode integrals derived in~\cite{Bode_1945, Middleton_1991, Seron_2012}, despite the difference of initial conditions,~\cite{Ishii_SCL_2011}. Moreover, when the signals are not stationary and the Bode-like integrals are not defined, we can still resort to mutual information rates to describe the performance limitations in stochastic continuous-time systems as in ~\cite{Martin_TAC_2007, Li_TAC_2013}, which we will elaborate in~\hyperref[sec4]{Section~4}.
\end{remark}

\noindent In the following sections, we first discuss the information-theoretic representation of the Bode-like integrals defined in~\eqref{Bode-Like}, and then we derive the lower bounds for these performance limitations.

\section{Information-Theoretic Representations of Bode-Like Integrals}\label{sec3}
When signals in feedback systems are stationary Gaussian, we show that the Bode-like integrals~\eqref{Bode-Like} can be bounded below by the difference between two mutual information rates. This information-theoretic representation of Bode-like integrals not only enables to derive the lower bounds of Bode-like integrals with tools from information theory, but provides with an alternative metric to measure the performance limitations of stochastic continuous-time systems when Bode-like integrals are undefined, \textit{i.e.} the stationary assumption fails to hold. The following theorem gives the information-theoretic representation for sensitivity and load disturbance sensitivity Bode-like integrals.

\begin{theorem}\label{thm1a}
	For the general feedback control system, when $(u, v)$ and $(w, v)$ form stationary processes, $\phi_u, \ \phi_v, \textrm{ and } \phi_w \in \mathbb{F}$, and $w$ is a stationary Gaussian process, the sensitivity Bode-like integral satisfies
	\begin{equation}\label{Thm1_R1}
		\dfrac{1}{2\pi}\int_{-\infty}^{\infty} \log T_{uw}(\omega) d\omega \geq I_{\infty}(u; v) - I_{\infty}(w; v),
	\end{equation}
	and the load disturbance sensitivity Bode-like integral satisfies
	\begin{equation}\label{Thm1_R2}
		\dfrac{1}{2\pi}\int_{-\infty}^{\infty} \log T_{yw}(\omega) {d\omega}  \geq I_{\infty}(y; v) - I_{\infty}(w; v) +  \dfrac{1}{2\pi}\int_{-\infty}^{\infty}\log|G(s)| {d\omega}.
	\end{equation}
\end{theorem}

\begin{proof}[\textbf{\textit{Proof}}]
	The block diagram of a general continuous-time feedback system subject to control noise $w(t)$ is illustrated in~\hyperref[fig4]{Figure~3}.
	\begin{figure}[H]\label{fig4}
		\centering \vspace{-1em}
		\includegraphics[width=0.5\textwidth]{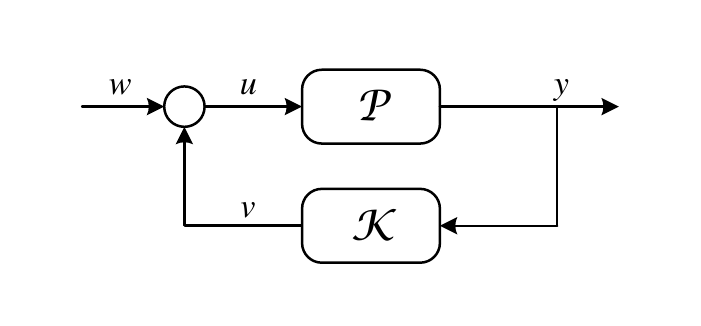} \vspace{-2em}
		\caption{\small General feedback system with control noise.}
	\end{figure}
	\noindent The first inequality in~\hyperref[thm1a]{Theorem~1} has been proved as \textit{Theorem~4.8} in~\cite{Li_TAC_2013}. Thus, only the proof of inequality~\eqref{Thm1_R2} is provided here. When ${y}_G, {u}_G$, and ${v}_G$ respectively denote the Gaussian processes with the same covariances as the random processes ${y}, {u}$, and ${v}$, the following inequality holds
	\begin{align}
			I(y^t; v^t) - I(w^t; v^t) \neweq{(a)} \ &  \lim_{k \rightarrow \infty} \left\{ I( [y^{(\delta(k))}]^k; [u^{(\delta(k))}]^k) - I( [d^{(\delta(k))}]^k; [u^{(\delta(k))}]^k)\right\} \nonumber \\
											   	 \neweq{(b)} \ & \lim_{k \rightarrow \infty} \Big\{  h([y^{(\delta(k))}]^k) - h([u^{(\delta(k))}]^k | [v^{(\delta(k))}]^k) - \mathbb{E}[\log |J_{[u^{(\delta(k))}]^k}( [y^{(\delta(k))}]^k)| ] \nonumber \\
											   						 &  \hspace{86pt} - h([w^{(\delta(k))}]^k) + h([w^{(\delta(k))}]^k | [v^{(\delta(k))}]^k) \Big\} \label{Thm1_eq1}\\
			\newleq{(c)} \ & \lim_{k \rightarrow \infty} \left\{  h([y_G^{(\delta(k))}]^k) - h([w^{(\delta(k))}]^k) - \mathbb{E} [\log |J_{[u^{(\delta(k))}]^k}([y^{(\delta(k))}]^k) | ]  \right\}  \nonumber \\
			\neweq{(d)} \ & h([y_G^{(\delta(k))}]^k) - h([u_G^{(\delta(k))}]^k | [v_G^{(\delta(k))}]^k) - \mathbb{E} [\log |J_{[u_G^{(\delta(k))}]^k}([y_G^{(\delta(k))}]^k)|] \nonumber \\
								 & \hspace{56pt} - h([w^{(\delta(k))}]^k) + h([w^{(\delta(k))}]^k| [v_G^{(\delta(k))}]^k)  \nonumber  \\
			\neweq{(e)} \ & I(y^t_G; v^t_G) - I(w^t; v^t_G). \nonumber
	\end{align}

\noindent In the derivations above (a) follows \hyperref[ppt0]{Property~1}; (b) is obtained by applying \hyperref[ppt1]{Property~2} and \hyperref[ppt2]{Property~3} with $J_{[u_G^{(\delta(k))}]^k}([y_G^{(\delta(k))}]^k)$ being the Jacobian matrix  of $[y_G^{(\delta(k))}]^k$ with respect to $[u_G^{(\delta(k))}]^k$; (c) employs \hyperref[ppt3]{Property~4} and identity $h([u^{(\delta(k))}]^k | [v^{(\delta(k))}]^k)  =  h([w^{(\delta(k))}]^k - [v^{(\delta(k))}]^k | [v^{(\delta(k))}]^k) =  h([w^{(\delta(k))}]^k | [v^{(\delta(k))}]^k)$ derived by \hyperref[ppt1]{Property~2}; (d) uses the identities $h([u_G^{(\delta(k))}]^k | [v_G^{(\delta(k))}]^k) = \break h([w^{(\delta(k))}]^k | [v_G^{(\delta(k))}]^k)$ derived by \hyperref[ppt1]{Property~2} and $|J_{[u^{(\delta(k))}]^k}( [y^{(\delta(k))}]^k)| = |J_{[u_G^{(\delta(k))}]^k}( [y_G^{(\delta(k))}]^k)|$; and (e) follows the same arguments in steps (a) and (b). Based on inequality~\eqref{Thm1_eq1} and derivations in~\cite{Li_TAC_2013}, when noise ${w}(t)$ is stationary Gaussian, we have the following relationship on mutual information rates: $I_{\infty}({y}; {v}) - I_{\infty}({w}; {v}) \leq I_{\infty}({y}_G; {v}_G) - I_{\infty}({w}; {v}_G)$, where equality holds when $y(t)$ is Gaussian. Since $({y}_G, {v}_G)$ and $({w}, {v}_G)$ form stationary Gaussian processes and $\phi_{{u}}, \phi_{{v}}, \textrm{and } \phi_{{w}} \in \mathbb{F}$, subject to~\hyperref[ppt4]{Property~5}, we have the following inequality between load disturbance sensitivity Bode-like integral and mutual information rates:
\begin{align}
	 	& I_{\infty}({y}; {v}) - I_{\infty}({w}; {v})\nonumber \\
		\leq \ & I_\infty({y}_G; {v}_G) - I_\infty({w}; {v}_G) \nonumber \\
		\neweq{(a)} \ & -\dfrac{1}{4\pi} \int_{-\infty}^{\infty} \log\left(1 - \dfrac{\phi_{{y}_G{v}_G}({\omega})\phi_{{v}_G{y}_G}({\omega})}{\phi_{{y}_G}({\omega}) \phi_{{v}_G}({\omega})}\right) d{\omega} + \dfrac{1}{4\pi} \int_{-\infty}^{\infty} \log\left(1 - \dfrac{\phi_{{w}{v}_G}({\omega})\phi_{{v}_G{w}}({\omega})}{\phi_{{w}}({\omega}) \phi_{{v}_G}({\omega})}\right) d{\omega} \nonumber \\
		\neweq{(b)} \ & \dfrac{1}{2\pi}  \int_{-\infty}^{\infty}  \log T_{{y}{w}}({\omega}) \ d{ \omega} + \dfrac{1}{4\pi}\int_{-\infty}^{\infty} \log \dfrac{\phi_{{u}_G}({\omega}) \phi_{{v}_G}({\omega}) - \phi_{{u}_G{v}_G}({\omega})\phi_{{v}_G{u}_G}({\omega})}{\phi_{{y}_G}({\omega})\phi_{{v}_G}({\omega}) - \phi_{{y}_G{v}_G}({\omega})\phi_{{v}_G{y}_G}({\omega})} d{\omega} \label{Thm1_eq2} \\
		\neweq{(c)} \ & \dfrac{1}{2\pi}\int_{-\infty}^{\infty} \log T_{{y}{w}}({\omega}) \ d{ \omega} - \dfrac{1}{2\pi}\int_{-\infty}^{\infty} \log |{G}({s})| d{\omega}. \nonumber
\end{align}

\noindent In these derivations (a) is obtained by applying~\hyperref[ppt4]{Property~5} to $I_\infty({y}_G, {v}_G)$ and $I_\infty({w}, {v}_G)$; (b) employs \hyperref[lem1]{Lemma~1} and identities, $\phi_{{w}} = \phi_{{v}_G} + \phi_{{v}_G{u}_G} + \phi_{{u}_G{v}_G} + \phi_{{u}_G}$, $\phi_{{w}{v}_G}= \phi_{{u}_G{v}_G} + \phi_{{v}_G}$, and $\phi_{{v}_G{w}}=  \phi_{{v}_G{u}_G} + \phi_{{v}_G}$, which can be derived by the definition of PSD in~\eqref{PSD} and identity  ${w} = {v}_G + {u}_G$; and (c) follows $\phi_{{u}_G} = \allowbreak {G}^{-1}({s}){G}^{-1}(-{s}) \phi_{{y}_G}$, $\phi_{{u}_G{v}_G} = {G}^{-1}(-{s}) \phi_{{y}_G{v}_G}$, and $\phi_{{v}_G{u}_G} = {G}^{-1}({s})\phi_{{v}_G{y}_G}$, which can be inferred from the proof of~\hyperref[lem1]{Lemma~1} in~\hyperref[app]{Appendix}. Inequality~\eqref{Thm1_eq2} readily implies inequality~\eqref{Thm1_R2} in~\hyperref[thm1a]{Theorem~1}. This completes the proof.
\end{proof}

\begin{remark}
	According to~\hyperref[def3]{Definition~3}, all rational functions belong to class $\mathbb{F}$ function~\cite{Li_TAC_2013}. A presumption of derivation in~\eqref{Thm1_eq2} is that $\phi_{{v}_G}, \ \phi_{{u}_G},$ and $\phi_{{y}_G} \in \mathbb{F}$. This presumption can be easily fulfilled, especially when signals are Gaussian, since it is well known that the PSD of a zero-mean stationary Gaussian signal over all frequencies is a constant,~\cite{Fang_2017}. Meanwhile, the employment of~\hyperref[ppt2]{Property~3} in~\eqref{Thm1_eq1} requires $y=f(t)$ in~\hyperref[fig4]{Figure~3} be an injective function, which simplifies the technical development while delivering the essential conceptual message and was previously assumed  also in~\cite{Okano_Auto_2009,  Fang_TAC_2017}.
\end{remark}

The following theorem gives the information-theoretic representations of noise sensitivity and complementary sensitivity Bode-like integrals in the presence of measurement noise $d$.

\begin{manualtheorem}{\ref*{thm1a}\'}\label{thm1b}
	For the general continuous-time feedback control system, when $(\tilde{y}, \tilde{e})$ and $(\tilde{d}, \tilde{e})$ form stationary processes, $\phi_{\tilde{y}}, \ \phi_{\tilde{e}}, \textit{ and } \phi_{\tilde{d}} \in \mathbb{F}$, and $\tilde{d}$ is a stationary Gaussian process, the complementary sensitivity and noise sensitivity Bode-like integrals satisfy
	\begin{equation}\label{thm1b_a}
	\dfrac{1}{2\pi}\int_{-\infty}^{\infty} \log T_{yd}(\omega) \ \dfrac{d\omega}{\omega^2} \geq I_{\infty}(\tilde{y}; \tilde{e}) - I_{\infty}(\tilde{d}; \tilde{e}),
	\end{equation}
	\begin{equation}\label{thm1b_b}
	\dfrac{1}{2\pi}\int_{-\infty}^{\infty} \log T_{ud}(\omega) \ \dfrac{d\omega}{\omega^2}  \geq 	I_{\infty}(\tilde{u}; \tilde{e}) - I_{\infty}(\tilde{d}; \tilde{e}) - \dfrac{1}{2\pi}\int_{-\infty}^{\infty}\log|G(s)| \dfrac{d\omega}{\omega^2}.
	\end{equation}
\end{manualtheorem}
\begin{proof}

\noindent Consider the inverse system subject to measurement noise shown in~\hyperref[fig7]{Figure~4}.
\begin{figure}[H]\label{fig7}
	\centering \vspace{-2em}
	\includegraphics[width=0.65\textwidth]{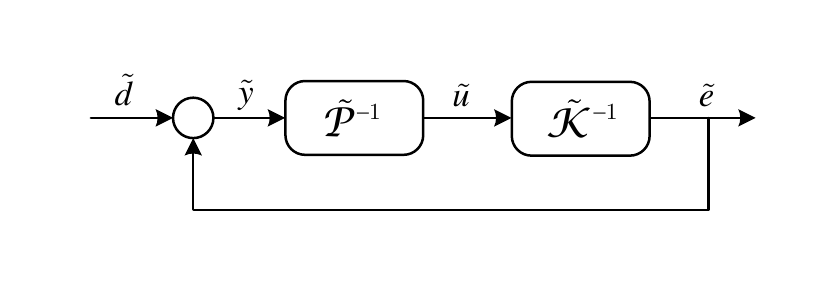} \vspace{-3.5em}
	\caption{\small Inverse system with measurement noise.}
\end{figure}
\noindent First, we consider the complementary sensitivity Bode-like integral. Preliminary work on this Bode-like integral was presented in an earlier paper~\cite{Wan_CDC_2018}. We again use $\tilde{y}_G$ and $\tilde{e}_G$ to respectively denote the Gaussian processes with the same covariance as the random processes $\tilde{y}$ and $\tilde{e}$. Then the following inequality can be established:

	\begin{align}
			I(\tilde{y}^t; \tilde{e}^t) - I(\tilde{d}^t; \tilde{e}^t) \neweq{(a)} \ & \lim_{k \rightarrow \infty}  \Big\{ h([\tilde{y}^{(\delta(k))}]^{k}) - h([\tilde{y}^{(\delta(k))}]^{k} | [\tilde{e}^{(\delta(k))}]^{k}) - h([\tilde{d}^{(\delta(k))}]^{k}) \nonumber\\
			& \hspace{200pt} + h([\tilde{d}^{(\delta(k))}]^{k} | [\tilde{e}^{(\delta(k))}]^{k}) \Big\} \nonumber \\
			\newleq{(b)} \ & \lim_{k \rightarrow \infty}  \left\{h([\tilde{y}_G^{(\delta(k))}]^{k}) - h([\tilde{d}^{(\delta(k))}]^{k}) \right\} \label{eq18} \\
			\neweq{(c)} \ & I(\tilde{y}^t_G; \tilde{e}^t_G) - I(\tilde{d}^t; \tilde{e}^t_G).\nonumber
	\end{align}
In these derivations 	(a) is obtained by \hyperref[ppt0]{Property~1} and \hyperref[ppt1]{Property~2}; (b) employs~\hyperref[ppt3]{Property~4} and identity $h([\tilde{y}^{(\delta(k))}]^{k} | \allowbreak [\tilde{e}^{(\delta(k))}]^{k}) = h([\tilde{d}^{(\delta(k))}]^{k} -[\tilde{e}^{(\delta(k))}]^{k} | [\tilde{e}^{(\delta(k))}]^{k}) = h([\tilde{d}^{(\delta(k))}]^{k} | [\tilde{e}^{(\delta(k))}]^{k})$ derived by \hyperref[ppt1]{Property~1}, and equality holds when $(\tilde{y}, \tilde{e})$ and $(\tilde{d}, \tilde{e})$ are Gaussian; (c) follows \hyperref[ppt0]{Property~1}, \hyperref[ppt1]{Property~2} and identity $h([\tilde{y}_G^{(\delta(k))}]^{k}|[\tilde{e}_G^{(\delta(k))}]^{k}) = h( [\tilde{d}^{(\delta(k))}]^{k} | [\tilde{e}_G^{(\delta(k))}]^{k})$. By~\eqref{eq18}, we readily have $I_\infty(\tilde{y}; \tilde{e}) - I_\infty(\tilde{d}; \tilde{e}) \leq I_\infty(\tilde{y}_G; \tilde{e}_G) - I_\infty(\tilde{d}; \tilde{e}_G)$. When $(\tilde{y}_G, \tilde{e}_G)$ and $(\tilde{d}, \tilde{e}_G)$ are stationary Gaussian processes and $\phi_{\tilde{y}_G}, \ \phi_{\tilde{e}_G}$, and $\phi_{\tilde{d}_G}\in\mathbb{F}$, by~\hyperref[ppt4]{Property~5}, an inequality between mutual information rates and complementary sensitivity Bode-like integral can be established:
	\begin{align}\label{thm1b_eq2}
	 &  I_\infty(\tilde{y}; \tilde{e}) - I_\infty(\tilde{d}; \tilde{e}) \nonumber\\
	\leq \ &	I_\infty(\tilde{y}_G; \tilde{e}_G) - I_\infty(\tilde{d}; \tilde{e}_G) \nonumber\\
	\neweq{(a)} \ &  -\dfrac{1}{4\pi} \int_{-\infty}^{\infty} \log\left(1 - \dfrac{\phi_{\tilde{y}_G\tilde{e}_G}(\tilde{\omega})\phi_{\tilde{e}_G\tilde{y}_G}(\tilde{\omega})}{\phi_{\tilde{y}_G}(\tilde{\omega}) \phi_{\tilde{e}_G}(\tilde{\omega})}\right) d\tilde{\omega} + \dfrac{1}{4\pi} \int_{-\infty}^{\infty} \log\left(1 - \dfrac{\phi_{\tilde{d}\tilde{e}_G}(\tilde{\omega})\phi_{\tilde{e}_G\tilde{d}}(\tilde{\omega})}{\phi_{\tilde{d}}(\tilde{\omega}) \phi_{\tilde{e}_G}(\tilde{\omega})}\right) d\tilde{\omega} \nonumber \\
	\neweq{(b)} \ &  \dfrac{1}{2\pi}\int_{-\infty}^{\infty} \log T_{\tilde{y}\tilde{d}}(\tilde{\omega}) \ d{\tilde \omega} + \dfrac{1}{4\pi} \int_{-\infty}^{\infty} \log \dfrac{\phi_{\tilde{y}_G}(\tilde{\omega})\phi_{\tilde{e}_G}(\tilde{\omega}) - \phi_{\tilde{y}_G\tilde{e}_G}(\tilde{\omega})\phi_{\tilde{e}_G\tilde{y}_G}(\tilde{\omega})}{\phi_{\tilde{y}_G}(\tilde{\omega})\phi_{\tilde{e}_G}(\tilde{\omega}) - \phi_{\tilde{y}_G\tilde{e}_G}(\tilde{\omega})\phi_{\tilde{e}_G\tilde{y}_G}(\tilde{\omega})} \ d\tilde{\omega}  \\
	\neweq{(c)} \ & \dfrac{1}{2\pi}\int_{-\infty}^{\infty} \log T_{yd}(\omega) \dfrac{d\omega}{\omega^2}.   \nonumber
	\end{align}
	\noindent In the derivations above (a) is the application of \hyperref[ppt4]{Property~5}; (b) follows~\hyperref[lem1]{Lemma~1} and identity $\tilde{d} = \tilde{y}_G + \tilde{e}_G$, which implies $\phi_{\tilde{d}} = \phi_{\tilde{y}_G} + \phi_{\tilde{y}_G \tilde{e}_G} + \phi_{\tilde{e}_G\tilde{y}_G} + \phi_{\tilde{e}_G}$, $\phi_{\tilde{d}\tilde{e}_G} = \phi_{\tilde{y}_G\tilde{e}_G} + \phi_{\tilde e_G}$, and $\phi_{\tilde{e}_G \tilde{d}} = \phi_{\tilde{e}_G \tilde{y}_G} + \phi_{\tilde{e}_G}$; and (c) can be implied by the frequency transformation~\eqref{freq_trans}.  Inequality~\eqref{thm1b_a} in \hyperref[thm1b]{Theorem~1'} can then be readily implied from~\eqref{thm1b_eq2}.

	Next, we consider the Bode integral of noise sensitivity-like function $T_{ud}(\omega)$. Since the steps for deriving inequality~\eqref{thm1b_b} can be inferred from the preceding proofs, we only show some critical steps in the following. Similar to the derivations in~\eqref{Thm1_eq1}, for the noise sensitivity Bode-like integral, we have the following inequality on mutual information rates: $	I_\infty(\tilde{u}_G; \tilde{e}_G) - I_\infty(\tilde{d}; \tilde{e}_G) \geq I_\infty(\tilde{u}; \tilde{e}) - I_\infty(\tilde{d}; \tilde{e})$,
%
	where $\tilde{u}_G$ and $\tilde{e}_G$ denote the Gaussian processes with the same covariances as $\tilde{u}$ and $\tilde{e}$, respectively. With the facts that $(\tilde{u}_G, \tilde{e}_G)$ and $(\tilde{d}, \tilde{e}_G)$ are stationary Gaussian processes and $\phi_{\tilde{u}_G}, \ \phi_{\tilde{e}_G}, \textrm{ and } \phi_{\tilde{d}} \in \mathbb{F}$, we have:
	\begin{align}\label{thm1b_eq4}
		&  I_\infty(\tilde{u}; \tilde{e}) - I_\infty(\tilde{d}; \tilde{e}) \nonumber\\
		\leq  \ &	I_\infty(\tilde{u}_G; \tilde{e}_G) - I_\infty(\tilde{d}; \tilde{e}_G) \nonumber\\
		\neweq{(a)} \ &  \dfrac{1}{2\pi}\int_{-\infty}^{\infty} \log T_{\tilde{u}\tilde{d}}(\tilde{\omega}) \ d{\tilde \omega} + \dfrac{1}{4\pi} \int_{-\infty}^{\infty} \log \dfrac{\phi_{\tilde{y}_G}(\tilde{\omega})\phi_{\tilde{e}_G}(\tilde{\omega}) - \phi_{\tilde{y}_G\tilde{e}_G}(\tilde{\omega})\phi_{\tilde{e}_G\tilde{y}_G}(\tilde{\omega})}{\phi_{\tilde{u}_G}(\tilde{\omega})\phi_{\tilde{e}_G}(\tilde{\omega}) - \phi_{\tilde{u}_G\tilde{e}_G}(\tilde{\omega})\phi_{\tilde{e}_G\tilde{u}_G}(\tilde{\omega})} \ d\tilde{\omega}  \\
		\neweq{(b)} \ & \dfrac{1}{2\pi}\int_{-\infty}^{\infty} \log T_{ud}(\omega) \dfrac{d\omega}{\omega^2} + \dfrac{1}{2\pi}\int_{-\infty}^{\infty}\log|G(s)| \dfrac{d\omega}{\omega^2}  \nonumber
	\end{align}
In these derivations 	(a) employs~\hyperref[ppt4]{Property~5}, \hyperref[lem1]{Lemma~1}, and identity $\tilde{d} = \tilde{y}_G + \tilde{e}_G$; and (b) follows~\eqref{freq_trans},~\eqref{tf_auxiliary}, and the identities, $\phi_{\tilde{u}_G} =  \tilde{G}^{-1}(\tilde{s})\tilde{G}^{-1}(-\tilde{s}) \phi_{\tilde{y}_G}$, $\phi_{\tilde{u}_G\tilde{e}_G} = \tilde{G}^{-1}(-\tilde{s}) \phi_{\tilde{y}_G\tilde{e}_G}$, and $\phi_{\tilde{e}_G\tilde{u}_G} = \tilde{G}^{-1}(\tilde{s})\phi_{\tilde{e}_G\tilde{y}_G}$, which can be inferred from the proof of~\hyperref[lem1]{Lemma~1} in~\hyperref[app]{Appendix}. Inequality~\eqref{thm1b_b} can then be readily implied from~\eqref{thm1b_eq4}. This completes the proof.
\end{proof}

\begin{remark}
	The terms $(2\pi)^{-1}\int_{-\infty}^{\infty}\log|G(s)| d\omega$ in~\eqref{Thm1_R2} and $(2\pi)^{-1} \int_{-\infty}^{\infty}\log|G(s)| / {\omega^2} \ {d\omega}$ in~\eqref{thm1b_b} are constants once a linear plant model $G(s)$ is given, thus the lower bounds in~\hyperref[thm1a]{Theorem~1} and~\hyperref[thm1b]{Theorem~1'} are invariant of the choice of control mapping. Nevertheless, these two integrals are not always bounded for arbitrary $G(s)$. A heuristic result on the boundedness of these types of integrals is available in~\cite{Wu_TAC_1992}.
\end{remark}

With the information-theoretic representations of Bode-like integrals derived in this section, we now can establish the lower bounds of Bode-like integrals by deriving the lower bounds of their information-theoretic representations with tools and preliminary results from information theory.

\section{Lower Bounds of Performance Limitations}\label{sec4}
Compared with the definitions of Bode-like integrals, which require the stationary Gaussian condition, the existence of their information-theoretic representations is less restrictive, as it only requires the stationary condition. In this section, the lower bounds of these information-theoretic representations, and hence the lower bounds of Bode-like integrals are derived for continuous-time stationary feedback systems. For systems subject to control noise, we have the following result.

\begin{theorem} \label{thm2a}
	When the closed-loop system shown in~\hyperref[fig4]{Figure~3} is mean-square stable, and the control noise $w(t)$ is stationary Gaussian, we have
		\begin{equation}\label{Thm2_R1}
			I_{\infty}(u; v) - I_{\infty}(w; v) \geq \sum_{p_i \in \mathcal{UP}} p_i,
		\end{equation}
		\begin{equation}\label{Thm2_R2}
		I_{\infty}({y}; {v}) - I_{\infty}({w}; {v}) \geq \sum_{p_i \in \mathcal{UP}} {p_i},
		\end{equation}
		where $\mathcal{UP}$ denotes the set of unstable poles in plant $\mathcal{P}$.
\end{theorem}

\begin{proof}
	The first inequality~\eqref{Thm2_R1} in~\hyperref[thm2a]{Theorem~2} has been proven as \textit{Theorem 4.8} in~\cite{Li_TAC_2013}. In order to derive the lower bound in inequality~\eqref{Thm2_R2}, consider the general feedback configuration illustrated in~\hyperref[fig4]{Figure~3}. We have the following relationship
	\begin{align}
		I({u}^t; {v}^t) & \neweq{(a)} \lim_{k \rightarrow \infty}\{ h([u^{(\delta(k))}]^k) - h([u^{(\delta(k))}]^k | [v^{(\delta(k))}]^k) \} \nonumber\\
		& \neweq {(b)} h([y^{(\delta(k))}]^k) - \mathbb{E}[|J_{[u^{(\delta(k))}]^k}([y^{(\delta(k))}]^k)|] - h([y^{(\delta(k))}]^k | [v^{(\delta(k))}]^k) + \mathbb{E}[|J_{[u^{(\delta(k))}]^k}([y^{(\delta(k))}]^k)|]\nonumber \\
		& \neweq {(c)} I({y}^t; {v}^t). \label{Thm2_eq2}
	\end{align}
In the above derivations (a) follows~\hyperref[ppt0]{Property~1} and~\hyperref[ppt1]{Property~2}; (b) employs \hyperref[ppt2]{Property~3} with the same injective assumption in~\eqref{Thm1_eq1}, which gives $h([u^{(\delta(k))}]^k) = h([y^{(\delta(k))}]^k) - \mathbb{E}[|J_{[u^{(\delta(k))}]^k}([y^{(\delta(k))}]^k)|]$ and $h([u^{(\delta(k))}]^k | [v^{(\delta(k))}]^k) \allowbreak = h([y^{(\delta(k))}]^k | [v^{(\delta(k))}]^k) - \mathbb{E}[|J_{[u^{(\delta(k))}]^k}([y^{(\delta(k))}]^k)|]$, where $J_{[u^{(\delta(k))}]^k}([y^{(\delta(k))}]^k)$ is the Jacobian matrix of vector $[y^{(\delta(k))}]^k$ with respect to vector $[u^{(\delta(k))}]^k$; and (c) follows~\hyperref[ppt0]{Property~1} and~\hyperref[ppt1]{Property~2}. The derivation in~\eqref{Thm2_eq2} can also be verified by data-processing inequality~\cite{Cover_2012}, and equation~\eqref{Thm2_eq2} implies $I_\infty({u}; {v}) = I_\infty({y}; {v})$, which combing with~\eqref{Thm2_R1} gives~\eqref{Thm2_R2} in~\hyperref[thm2a]{Theorem~2}. This completes the proof.
\end{proof}

The following theorem derives the lower bounds of the information-theoretic representations for the feedback systems subject to measurement noise described in~\hyperref[fig7]{Figure~4}.

\begin{manualtheorem}{\ref*{thm2a}\'}\label{thm2b}
	When the closed-loop system shown in~\hyperref[fig7]{Figure~4} is mean-square stable and the inverse frequency measurement noise $\tilde{d}(t)$ is stationary Gaussian, we have
	\begin{equation}\label{Thm2_R3}
			I_{\infty}(\tilde{y}; \tilde{e}) - I_{\infty}(\tilde{d}; \tilde{e}) \geq \sum_{z_i \in \mathcal{UZ}}\dfrac{1}{z_i},
	\end{equation}
	\begin{equation}\label{Thm2_R4}
			I_{\infty}(\tilde{u}; \tilde{e}) - I_{\infty}(\tilde{d}; \tilde{e}) \geq \sum_{z_i \in \mathcal{UZ}}\dfrac{1}{z_i},
	\end{equation}
	where $\mathcal{UZ}$ denotes the set of nonminimum phase zeros of plant $\mathcal{P}$.
\end{manualtheorem}
\begin{proof}
	Consider the inverse system shown in~\hyperref[fig7]{Figure~4}. Applying \textit{Theorem 4.8} from~\cite{Li_TAC_2013} and noticing that the poles of the inverse plant model $\tilde{\mathcal{P}}^{-1}$ are at $s = 1/z_i$, the first inequality~\eqref{Thm2_R3} can be obtained~\cite{Wan_CDC_2018}. In order to derive inequality~\eqref{Thm2_R4}, consider the following equations
	\begin{equation}\label{thm2b_eq1}
		\begin{split}
			I(\tilde{y}^{t};\tilde{e}^t)  &  \neweq{(a)} \lim_{k \rightarrow \infty} \left\{h([\tilde{y}^{(\delta(k))}]^k) - h([\tilde{y}^{(\delta(k))}]^k|[\tilde{e}^{(\delta(k))}]^k) \right\} \\
			& \neweq{(b)}  \lim_{k \rightarrow \infty} \Big\{ h([\tilde{u}^{(\delta(k))}]^k) - \mathbb{E}[|J_{[\tilde{y}^{(\delta(k))}]^k}([\tilde{u}^{(\delta(k))}]^k)|] - h([\tilde{u}^{(\delta(k))}]^k | [\tilde{e}^{(\delta(k))}]^k)\\
			& \hspace{224pt} + \mathbb{E}[|J_{[\tilde{y}^{(\delta(k))}]^k}([\tilde{u}^{(\delta(k))}]^k)|] \Big\} \\
			& \neweq{(c)} I(\tilde{u}; \tilde{e})
			\end{split}
	\end{equation}
Here	(a) applies~\hyperref[ppt0]{Property~1} and~\hyperref[ppt1]{Property~2}; (b) adopts the assumption that $\tilde{u} = \tilde{f}^{-1}(\tilde{y})$ in~\hyperref[fig7]{Figure~4} is injective and \hyperref[ppt2]{Property~3}, which give $h([\tilde{y}^{(\delta(k))}]^k) = h([\tilde{u}^{(\delta(k))}]^k)  - \mathbb{E}[|J_{[\tilde{y}^{(\delta(k))}]^k}([\tilde{u}^{(\delta(k))}]^k)|]$ and $h([\tilde{y}^{(\delta(k))}]^k | [\tilde{e}^{(\delta(k))}]^k) \allowbreak = h([\tilde{u}^{(\delta(k))}]^k | [\tilde{e}^{(\delta(k))}]^k)  - \mathbb{E}[|J_{[\tilde{y}^{(\delta(k))}]^k}([\tilde{u}^{(\delta(k))}]^k)|]$, where $J_{[\tilde{y}^{(\delta(k))}]^k}([\tilde{u}^{(\delta(k))}]^k)$ is the Jacobian matrix of vector $[\tilde{u}^{(\delta(k))}]^k$ with respect to vector $[\tilde{y}^{(\delta(k))}]^k$; and (c) follows~\hyperref[ppt0]{Property~1} and~\hyperref[ppt1]{Property~2}. Since \eqref{thm2b_eq1} indicates thats $I_\infty(\tilde{y}; \tilde{e}) = I_\infty(\tilde{u}; \tilde{e})$, combing it with~\eqref{Thm2_R3} gives inequality~\eqref{Thm2_R4}. This completes the proof.
\end{proof}

With all the preceding theorems, the following corollary gives the lower bounds of the Bode-like integrals in stochastic continuous-time systems.

\begin{corollary} \label{cor3}
	For continuous-time feedback control system that is closed-loop stable, when $(u, v)$ and $(w, v)$ form stationary processes, $\phi_u, \ \phi_v, \textrm{ and } \phi_w \in \mathbb{F}$, and $w$ is a stationary Gaussian process, the sensitivity and the load disturbance sensitivity Bode-like integrals satisfy
	\begin{equation}\label{cor_eq1}
	\dfrac{1}{2\pi}\int_{-\infty}^{\infty} \log T_{uw}(\omega) d\omega \geq \sum_{p_i \in \mathcal{UP}} p_i,
	\end{equation}
	\begin{equation}\label{cor_eq2}
	\dfrac{1}{2\pi}\int_{-\infty}^{\infty} \log T_{yw}(\omega) {d\omega} \geq \sum_{p_i \in \mathcal{UP}} {p_i} + \dfrac{1}{2\pi}\int_{-\infty}^{\infty}\log|G(s)| {d\omega}.
	\end{equation}
	When $(\tilde{y}, \tilde{e})$ and $(\tilde{d}, \tilde{e})$ form stationary processes, $\phi_{\tilde{y}}, \ \phi_{\tilde{e}}, \textit{ and } \phi_{\tilde{d}} \in \mathbb{F}$, and $\tilde{d}$ is a stationary Gaussian process, the complementary sensitivity and noise sensitivity Bode-like integrals satisfy
	\begin{equation}\label{cor_eq3}
	\dfrac{1}{2\pi}\int_{-\infty}^{\infty} \log T_{yd}(\omega) \ \dfrac{d\omega}{\omega^2} \geq \sum_{z_i \in \mathcal{UZ}}\dfrac{1}{z_i},
	\end{equation}
	\begin{equation}\label{cor_eq4}
	\dfrac{1}{2\pi}\int_{-\infty}^{\infty} \log T_{ud}(\omega) \ \dfrac{d\omega}{\omega^2}  \geq  \sum_{z_i \in \mathcal{UZ}}\dfrac{1}{z_i} - \dfrac{1}{2\pi}\int_{-\infty}^{\infty}\log|G(s)| \dfrac{d\omega}{\omega^2}.
	\end{equation}
\end{corollary}
\begin{proof}
	\hyperref[cor3]{Corollary~3} can be implied by applying \hyperref[thm2a]{Theorem~2} and \hyperref[thm2b]{Theorem~2'} to \hyperref[thm1a]{Theorem~1} and \hyperref[thm1b]{Theorem~1'}, respectively.
\end{proof}

\begin{remark}\label{rem6}
	From the proof of \hyperref[lem1]{Lemma~1} given in~\hyperref[app]{Appendix}, the following identities can be inferred, $T_{yw}(\omega) = G(s) T_{uw}(\omega)$ and $T_{yd}(\omega) = G(s) T_{ud}(\omega)$.	Once we admit the results given in~\eqref{cor_eq1} and~\eqref{cor_eq3}, we can also retrieve inequalities~\eqref{cor_eq2} and~\eqref{cor_eq4} by substituting the preceding two identities into~\eqref{cor_eq1} and~\eqref{cor_eq3}, respectively. Meanwhile, since a relationship similar to $T_{uw} + T_{yd} = 1$ does not exist between $T_{yw}$ and $T_{ud}$, the boundedness of right-hand side terms in~\eqref{cor_eq2} and~\eqref{cor_eq4} is more difficult to guarantee, and it is more proper to consider the lower bounds in~\eqref{cor_eq2} and~\eqref{cor_eq4} as a metric for performance limitations. Lastly, the injective assumptions adopted in~\eqref{Thm1_eq1}, \eqref{Thm2_eq2}, and~\eqref{thm2b_eq1} as well as the stationary Gaussian condition on the inverse frequency signals when studying the weighted Bode-like integrals are still some interesting topics that deserve further investigation.
\end{remark}

We notice that since the pioneering papers~\cite{Martin_TAC_2007, Martin_TAC_2008}, information-theoretic approaches have been widely employed to seek the lower bounds of Bode-like integrals. Among these papers, the lower bound of sensitivity Bode-like integral for discrete-time system has been discussed in~\cite{Martin_TAC_2007, Martin_TAC_2008}. The lower bound of complementary sensitivity Bode-like integral has been investigated in~\cite{Okano_Auto_2009}. A comprehensive investigation of the MIMO discrete-time control systems as well as the lower bounds for the discrete-time load disturbance sensitivity and noise sensitivity Bode-like integrals have been put forward in~\cite{Ishii_SCL_2011}. For  continuous-time control systems, \cite{Li_TAC_2013} has defined the lower bound of sensitivity Bode-like integral, and the lower bounds of continuous-time complementary sensitivity, load disturbance sensitivity, and noise sensitivity Bode-like integrals have been discussed in this paper. With information-theoretic approaches, Bode-like integrals of non-Gaussian system~\cite{Fang_TAC_2017}, continuous-time system with non-LTI plant~\cite{Fang_2017}, stochastic switched system~\cite{Li_Auto_2013}, and distributed system~\cite{Zhao_Auto_2014} have also been investigated.

\section{Conclusion}\label{sec5}
In this paper, we investigated the performance limitations of linear continuous-time control systems subject to control and measurement noise via an information-theoretic approach. Bode integrals of four different sensitivity-like functions were defined, and the relationship between Bode integrals and Bode-like integrals were established for stochastic continuous-time systems. The information-theoretic representations of Bode-like integrals were derived, and the lower bounds of these representations and hence the Bode-like integrals were established in terms of the unstable zeros and poles of plant model. Some open problems and challenges are discussed towards the end, and the hope is that more innovative results can be put forward to expand the frontier in this direction.

\section{Acknowledgment}
This work was partially supported by AFOSR and NSF. The authors would specially acknowledge the readers and staff on arXiv.org.

\section*{Appendix}\label{app}
\noindent \textbf{\textit{Proof of \hyperref[lem1]{Lemma 1}}}

Part of the proof in this appendix relies on the results given in~\cite{Astrom_2012}. We first consider the scenario in the presence of control noise $w(t)$ and start with Bode-like integrals. When both the plant model $\mathcal{P}$ and control mapping $\mathcal{K}$ in~\hyperref[fig1]{Figure~1} are linear, using $v(t)= (c*(g*u))(t) =  \int_{0}^{\infty} c(\theta) \left[ \int_{0}^{\infty} g(\eta)  u(t - \theta - \eta) d\eta \right]  d \theta$, we define $L(s) = G(s)C(s) = G(j\omega)C(j\omega) = \int_{0}^{\infty}(g*c)(t)\cdot\mathrm{e}^{-j\omega t} dt = \int_{0}^{\infty} l(t) \cdot \mathrm{e}^{-j\omega t} dt$. Since $w(t) = u(t) + v(t)$, the PSD function $\phi_d(\omega)$ in~\eqref{PSD} satisfies
\begin{equation}\label{App_Eq1}
	\begin{split}
		\phi_w(\omega)  =  \int_{-\infty}^{\infty} r_w(\tau) \cdot {\rm e}^{-j\omega\tau} d\tau &  = \int_{-\infty}^{\infty} \left[r_u(\tau) + r_{uv}(\tau) + r_{vu}(-\tau) + r_v(\tau)\right] {\rm e}^{-j\omega\tau} d\tau\\
		& = \phi_u(\omega) + \phi_{uv}(\omega) + \phi_{vu}(\omega) + \phi_v(\omega).
	\end{split}
\end{equation}
When $w(t)$, $v(t)$, and $u(t)$ are wide sense stationary, with $v(t) = \int_{0}^{\infty}l(\sigma') \cdot u(t - \sigma') d\sigma'$ and $\tau = \sigma - t$, the covariances $r_u, r_{uv}, r_{vu}$ and $r_v$ in~\eqref{App_Eq1} respectively satisfy $r_u(\sigma, t) = \textrm{Cov}[u(t+ \sigma -t), u(t)] =\textrm{Cov}[u(t+\tau), u(t)]= r_u(\tau)$, $	r_{uv}(\sigma, t) = \textrm{Cov}[u(\sigma), \int_{0}^{\infty}l(\sigma')u(t - \sigma') d\sigma']=\int_{0}^{\infty}l(\sigma') \cdot \textrm{Cov}[u(\sigma), u(t - \sigma')] d\sigma' = \int_{0}^{\infty}l(\sigma') \cdot r_u(\sigma' + \sigma - t) dv' = \int_{0}^{\infty}l(\sigma') \cdot r_u(\sigma' + \tau ) d\sigma' = r_{uv}(\tau)$, $r_{vu}(\sigma,t) = \textrm{Cov}[\int_{0}^{\infty}  l(\sigma') \cdot u(\sigma-\sigma')d\sigma', u(t)] = \int_{0}^{\infty}l(\sigma') \cdot \textrm{Cov}[u(\sigma-\sigma'), u(t)] d\sigma' = \int_{0}^{\infty}l(\sigma') \cdot r_u(- \sigma' + \sigma - t)d\sigma' = \int_{0}^{\infty} l(\sigma') \cdot r_u(-\sigma' + \tau) d\sigma'  = r_{vu}(\tau)$, and $r_v(\sigma,t) = \textrm{Cov}[\int_{0}^{\infty}l(\sigma') \cdot u(\sigma - \sigma')d\sigma', \int_{0}^{\infty}l(t') \cdot u(t-t')dt'] = \int_{0}^{\infty}\int_{0}^{\infty}l(\sigma') \cdot l(t') \cdot \textrm{Cov}[u(\sigma-\sigma'), u(t-t')] d\sigma' dt' = \int_{0}^{\infty}\int_{0}^{\infty}l(\sigma') \cdot l(t') \cdot r_u(\tau - \sigma' + t') d\sigma' dt' = r_v(\tau)$. Hence, the spectral density functions $\phi_{uv}, \phi_{vu}$, and $\phi_v$ in~\eqref{App_Eq1} respectively satisfy
	\begin{align}
	\phi_{uv}(\omega) & = \frac{1}{2\pi} \int_{-\infty}^{\infty} r_{uv}(\tau) \cdot {\rm e}^{-j\omega\tau}d\tau = \frac{1}{2\pi} \int_{-\infty}^{\infty} \mathrm{e}^{-j\omega\tau} \int_{0}^{\infty}l(\sigma') \cdot r_u(\sigma' + \tau) d\sigma' d\tau \nonumber \\
		& = \frac{1}{2\pi} \int_{-\infty}^{\infty} \mathrm{e}^{-j\omega(\tau+\sigma')} \int_{0}^{\infty} \mathrm{e}^{j\omega \sigma'} \cdot l(\sigma') \cdot r_u(\tau + \sigma') d\sigma' d\tau \\
		&= \frac{1}{2\pi} \int_{0}^{\infty}  {\rm e}^{j\omega \sigma'} \cdot l(\sigma') \int_{-\infty}^{\infty} {\rm e}^{-j\omega(\tau+\sigma')} \cdot r_u(\tau + \sigma') d\tau d\sigma' = L(-j\omega) \cdot \phi_u(\omega), \nonumber\\
	\phi_{vu}(\omega) &= \frac{1}{2\pi} \int_{-\infty}^{\infty} r_{vu}(\tau) \cdot {\rm e}^{-j\omega\tau}d\tau = \frac{1}{2\pi} \int_{-\infty}^{\infty} \mathrm{e}^{-j\omega(\tau-\sigma')} \int_{0}^{\infty}\mathrm{e}^{-j\omega \sigma'}\cdot l(\sigma') \cdot r_u(\tau - \sigma') d\sigma' d\tau \nonumber \\
	& = \frac{1}{2\pi} \int_{0}^{\infty}  {\rm e}^{-j\omega \sigma'} \cdot l(\sigma')  \int_{-\infty}^{\infty}  {\rm e}^{-j\omega(\tau-\sigma')} \cdot r_u(\tau - \sigma') d\tau d\sigma' = L(j\omega) \cdot \phi_u(\omega), \\
	\phi_v(\omega) & = \dfrac{1}{2\pi} \int_{-\infty}^{\infty} r_v(\tau) \cdot {\rm e}^{-j\omega\tau} d\tau = \dfrac{1}{2\pi} \int_{-\infty}^{\infty} \int_{0}^{\infty}\int_{0}^{\infty}l(\sigma') \cdot l(t') \cdot r_u(\tau - \sigma' + t') \ d\sigma' dt' \mathrm{e}^{-j\omega\tau}d\tau\nonumber\\
	& = \dfrac{1}{2\pi} \int_{0}^{\infty} {\rm e}^{j\omega t'} \cdot l(t') \int_{0}^{\infty} {\rm e}^{-j\omega \sigma'} \cdot l(\sigma')  \int_{-\infty}^{\infty} {\rm e}^{-j\omega(\tau-\sigma'+t')} \cdot r_u(\tau - \sigma' + t') \ d\tau d\sigma' dt' \label{App_Eq2}\\
	& = L(-j\omega) \cdot L(j\omega) \cdot \phi_u(\omega).  \nonumber
	\end{align}
Substituting~(\ref{App_Eq1})-(\ref{App_Eq2}) into the sensitivity-like function $T_{uw}(\omega)$ defined in~\eqref{SF2}, we can rewrite the sensitivity-like function as ${T}_{uw}(\omega) = [{\phi_u(\omega)} / {\phi_w(\omega)}]^{1/2} = \{\phi_u(\omega) / [\phi_u(\omega) + \phi_{uv}(\omega) + \phi_{vu}(\omega) + \phi_v(\omega)]\}^{1/2} = \{ \phi_u(\omega) / [(1+L(-j\omega)) \cdot (1+L(j\omega)) \cdot \phi_u(\omega)] \}^{1/2}$. When $\phi_u(\omega) \not\equiv 0$, with $L(s) = G(s) \cdot C(s)$, we have ${T}_{uw}(\omega)  = \sqrt{T_{uw}(-s) \cdot T_{uw}(s)}$. Since $T_{uw}(-s) = \bar{T}_{uw}(s)$, where $\bar{T}_{uw}(s)$ is the complex conjugate of ${T}_{uw}(s)$, the equality~\eqref{lem1a} in~\hyperref[lem1]{Lemma~1} can be retrieved from $(2\pi)^{-1} \cdot \int_{-\infty}^{\infty} \log {T}_{uw}(\omega) \ d\omega = (4\pi)^{-1} \int_{-\infty}^{\infty} \log \left[\bar{T}_{uw}(j\omega) \cdot {T}_{uw}(j\omega)\right] d \omega = (4\pi)^{-1} \int_{-\infty}^{\infty} \log \left|T_{uw}(j\omega)\right|^2 d \omega  = \break  (2\pi)^{-1} \cdot \int_{-\infty}^{\infty} \log | T_{uw}(j\omega) | d\omega$.

 Since $y(t) = g*u(t) = \int_{0}^{\infty} g(\theta) \cdot u(t - \theta) d\theta$, the auto-covariance of signal $y$ satisfies $r_y(\sigma, t) = \textrm{Cov}[\int_{0}^{\infty}g(\sigma') \cdot u(\sigma - \sigma')d\sigma', \int_{0}^{\infty}g(t') \cdot u(t-t')dt'] = \int_{0}^{\infty}\int_{0}^{\infty} g(\sigma') \cdot g(t') \cdot \textrm{Cov}[u(\sigma - \sigma'), u(t - t')] \allowbreak d\sigma' d t' = \int_{0}^{\infty}\int_{0}^{\infty} g(\sigma') \cdot g(t') \cdot r_u(\tau - \sigma' + t') d\sigma' dt' = r_y(\tau)$. Hence, the PSD of the stationary signal $y$ is
\begin{align}\label{App_Eq3}
	\phi_y(\omega) & = \dfrac{1}{2\pi}\int_{0}^{\infty}r_y(\tau) \cdot {\rm e}^{-j\omega\tau} d\tau  = \dfrac{1}{2\pi} \int_{-\infty}^{\infty} \int_{0}^{\infty}\int_{0}^{\infty} g(\sigma') \cdot g(t') \cdot r_u(\tau - \sigma' + t') d\sigma' dt'  \cdot {\rm e}^{-j\omega \tau} d\tau \nonumber\\
	& = \dfrac{1}{2\pi}\int_{0}^{\infty} \mathrm{e}^{j\omega t'} \cdot g(t') \int_{0}^{\infty} \mathrm{e}^{-j\omega\sigma'} \cdot g(\sigma') \int_{-\infty}^{\infty} \mathrm{e}^{-j\omega(\tau - \sigma' + t')}\cdot r_u(\tau - \sigma' + t') d\tau d\sigma' dt'\\
	& = G(-j\omega)\cdot G(j\omega) \cdot \phi_u(\omega). \nonumber
\end{align}
\noindent Substituting~(\ref{App_Eq1})-(\ref{App_Eq3}) into the load disturbance sensitivity-like function $T_{yw}(\omega)$ defined in~\eqref{SF2}, we can rewrite the load disturbance sensitivity-like function as follows ${T}_{yw}(\omega) = [{\phi_y(\omega)} / {\phi_w(\omega)}]^{1/2} = \{ \phi_y(\omega)  / [\phi_u(\omega) + \phi_{uv}(\omega) + \phi_{vu}(\omega) + \phi_v(\omega)] \}^{1/2} = \{ [G(-j\omega) \cdot G(j\omega) \cdot \phi_u(\omega)]  / [(1+L(-j\omega))\cdot(1 + L(j\omega)) \cdot \phi_u(\omega)]  \}^{1/2}$. When $\phi_u(\omega) \not\equiv 0$, it follows that ${T}_{yw}(\omega) = \sqrt{T_{yw}(-s) \cdot T_{yw}(s)}$. Since $T_{yw}(-j\omega) = \bar{T}_{yw}(j\omega)$, where $\bar{T}_{yw}(j\omega)$ is the complex conjugate of ${T}_{yw}(j\omega)$, the equality~\eqref{lem1b} in~\hyperref[lem1]{Lemma~1} can be retrieved from $(2\pi)^{-1} \int_{-\infty}^{\infty} \log {T}_{yw}(\omega) {d\omega}  = (4\pi)^{-1} \int_{-\infty}^{\infty} \log [\bar{T}_{yw}(j\omega) \cdot \break T_{yw}(j\omega)] {d\omega} =  (4\pi)^{-1}  \int_{-\infty}^{\infty} \log |T_{yw}(j\omega)|^2 {d\omega} =  (2\pi)^{-1}  \int_{-\infty}^{\infty} \log |T_{yw}(j\omega)| {d\omega}$. The steps for deriving~\eqref{lem1c} and~\eqref{lem1d} are similar to the preceding derivations; hence, only abbreviated steps are given.

Next, we consider the scenario in the presence of measurement noise $d(t)$. When both the plant model $\mathcal{P}$ and the control mapping $\mathcal{K}$ are linear, we have $y(t)= (c*(g*e))(t)   =   \int_{0}^{\infty} c(\theta) [ \int_{0}^{\infty} g(\eta)   \cdot \allowbreak  e(t - \theta - \eta) d\eta ]  d \theta = \int_{0}^{\infty}l(\sigma') \cdot  e(t - \sigma') d\sigma'$. Since $d(t) = e(t) + y(t)$, similar to the result presented in~\eqref{App_Eq1}, we have $\phi_d(\omega) = \phi_e(\omega) + \phi_{ey}(\omega) + \phi_{ye}(\omega) + \phi_{y}(\omega)$. When disturbance $d(t)$ is zero-mean stationary, with $y(t) = \int_{0}^{\infty}l(\sigma') \cdot e(t - \sigma') d\sigma'$, $u(t) =c * e(t) = \int_{0}^{\infty} c(\theta) \cdot e(t - \theta) d\theta $, and $\tau = \sigma - t$, the covariances $r_e, r_{ey}, r_{ye}$, $r_y$, and $r_u$ satisfy $r_{e}(\sigma, t) = r_e(\sigma - t) = r_e(\tau)$, $r_{ey}(\sigma, t) = \int_{0}^{\infty} l(\sigma') \cdot r_e(\sigma' + \tau) d\sigma' = r_{ey}(\tau)$, $r_{ye}(\sigma, t) = \int_{0}^{\infty} l(\sigma') \cdot r_e(-\sigma' + \tau) d\sigma' = r_{ye}(\tau)$, $r_y(\sigma, t) = \int_{0}^{\infty}\int_{0}^{\infty}l(\sigma') \cdot l(t') \cdot r_e(\tau - \sigma' + t') d\sigma' dt' = r_y(\tau)$, and $r_u(\sigma, t) = \int_{0}^{\infty}\int_{0}^{\infty} c(\sigma') \cdot c(t') \cdot r_e(t - \sigma' + t') d\sigma' dt' = r_u(\tau)$.
Hence, similar to~\eqref{App_Eq2}, the PSDs $\phi_{u}(\omega), \phi_{y}(\omega), \phi_{ey}(\omega)$, and $\phi_{ye}(\omega)$ satisfy $\phi_{u}(\omega)  = C(-j\omega) \cdot C(j\omega) \cdot \phi_{e}(\omega)$, $\phi_{y}(\omega)  = L(-j\omega) \cdot L(j\omega) \cdot \phi_e(\omega)$, $\phi_{ey}(\omega) = L(-j\omega) \cdot \phi_e(\omega)$, and $\phi_{ye}(\omega) = L(j\omega) \cdot \phi_e(\omega)$. Then when $\phi_e(\omega) \not\equiv 0$, we have the following relationship $T_{ud}(\omega) = [\phi_u(\omega) / \phi_d(\omega)]^{1/2} =\{ [C(-j\omega) \cdot C(j\omega) \cdot \phi_e(\omega)] / [(1+L(-j\omega))\cdot(1 + L(j\omega)) \cdot \phi_e(\omega)] \}^{1/2} = |T_{ud}(j\omega)|$, and $T_{yd}(\omega) = [\phi_y(\omega) / \phi_d(\omega)]^{1/2} = \{ [L(-j\omega) \cdot L(j\omega) \cdot \phi_e(\omega)] / [(1+L(-j\omega))\cdot(1 + L(j\omega)) \cdot \phi_e(\omega)] \}^{1/2} = \break |T_{yd}(j\omega)|$, which readily imply~\eqref{lem1c} and~\eqref{lem1d} in~\hyperref[lem1]{Lemma~1}. This completes the proof. \hfill \qed

\bibliographystyle{unsrt}
\bibliography{ref}

\end{document}